\documentclass[copyright,creativecommons]{eptcs}
\providecommand{\event}{GaM 2015}
\usepackage[USenglish]{babel}
\usepackage{graphicx}
\usepackage{float}
\usepackage{xspace}
\usepackage{color}
\usepackage[table]{xcolor}
\usepackage{hhline}
\usepackage{multirow}
\usepackage{cite}
\usepackage{courier}
\usepackage{multicol}
\usepackage{amssymb}
\usepackage{ifthen}
\usepackage{dsfont}
\usepackage{mattens}
\usepackage{mathtools}
\usepackage{pifont}
\usepackage[noend]{algorithmic}
\usepackage{wrapfig}
\usepackage{stmaryrd}
\usepackage{enumitem}

\floatstyle{ruled}
\newfloat{algorithm}{tbph}{al}
\floatname{algorithm}{Algorithm}

\usepackage[pdflatex=true,recompilepics=false]{gastex}
\DeclareMathAlphabet{\mathcal}{OMS}{cmsy}{m}{n}
\renewcommand{\algorithmiccomment}[1]{\hskip3em // #1} 

\newtheorem{definition}{Definition}
\newtheorem{lemma}{Lemma}
\newenvironment{proof}{\emph{Proof.}}{\hfill$\square$}

\newcommand{\noop}[1]{}

\newlength\savedintextsep
\newcommand{\Language}[1]{#1\xspace}

\newcommand{\LTS}[0]{\Language{LTS}}
\newcommand{\LTSs}[0]{\Language{LTSs}}

\newcommand{\LTSfunc}{\textrm{LTS}}
\newcommand{\cLTS}[1]{\ensuremath{\LTSfunc_{#1}}}

\newcommand{\graph}{\ensuremath{\mathcal{G}}\xspace}
\newcommand{\conflict}{\ensuremath{\mathcal{C}}\xspace}
\newcommand{\Xgraph}{\ensuremath{\mathcal{X}}\xspace}

\newcommand{\states}[1]{\ensuremath{{\mathcal{S}_{#1}}}\xspace}
\newcommand{\transitions}[1]{\ensuremath{{\mathcal{T}_{#1}}}\xspace}
\newcommand{\actions}[1]{\ensuremath{{\mathcal{A}_{#1}}}\xspace}

\newcommand{\ltstuple}[1]{\ensuremath{\langle \states{#1}, \actions{#1}, \transitions{#1} \rangle}\xspace}

\newcommand{\vectornot}[1]{\bS*{#1}}

\newcommand{\bisimilar}[1][]{%
  \setbox0=\hbox{\kern-.1ex{$\leftrightarrow$}\kern-.1ex}
  \setbox1=\vbox{\hbox{\raise .1ex \box0}\hrule}%
  \ensuremath{\mathrel{\hbox{\kern.1ex\box1\kern.1ex}_{#1}}}
}

\newcommand{\Llts} {\ensuremath{\mathcal{L}}}
\newcommand{\Rlts} {\ensuremath{\mathcal{R}}}
\newcommand{\Klts}{\ensuremath{\mathcal{K}}}
\newcommand{\graphresult}{\ensuremath{\mathcal{H}}}

\newcommand{\rz}{\ensuremath{{r_0}}}
\newcommand{\ro}{\ensuremath{{r_1}}}
\newcommand{\tr}{\ensuremath{{\it tr}}}

\newcommand{\bigo} {\ensuremath{\mathcal{O}}}

\def\authorrunning{A.J. Wijs}
\def\titlerunning{Confluence Detection for Transformations of Labelled Transition Systems}

\begin{document}

\pagestyle{plain}
\title{Confluence Detection for Transformations of Labelled Transition Systems}
\author{Anton Wijs
\institute{
 Department of Mathematics and Computer Science\\
 Eindhoven University of Technology\\
 P.O. Box 513, 5600 MB, Eindhoven, The Netherlands\\
 \email{A.J.Wijs@tue.nl}
}
}

\maketitle

\begin{abstract}
The development of complex component software systems can be made more manageable by first creating an abstract model and then incrementally adding details. Model transformation is an approach to add such details in a controlled way. In order for model transformation systems to be useful, it is crucial that they are confluent, i.e.\ that when applied on a given model, they will always produce a unique output model, independent of the order in which rules of the system are applied on the input. In this work, we consider Labelled Transition Systems (LTSs) to reason about the semantics of models, and LTS transformation systems to reason about model transformations. In related work, the problem of confluence detection has been investigated for general graph structures. We observe, however, that confluence can be detected more efficiently in special cases where the graphs have particular structural properties. In this paper, we present a number of observations to detect confluence of LTS transformation systems, and propose both a new confluence detection algorithm and a conflict resolution algorithm based on them.
\end{abstract}


\section{Introduction}
\label{sec:intro}

In Model-Driven Software Development, model transformation is a well-known technique to incrementally construct complex, often concurrent systems through manageable steps. It allows reasoning about a system at a high level of abstraction, and incrementally adding more information until a model has been constructed from which source code can be automatically derived. Some transformations add details or components to an existing model of a system under development, others refactor a model to make it easier to interpret, or translate a model to one written in a different modelling language. To reason about model transformations, often graph transformation is chosen as the underlying mechanism~\cite{graphformodeltrans,mtgt.comparative}.

It is crucial, though, that transformations are verifiable, i.e.\ that the definitions of transformations can be qualitatively analysed. Much work has been done on verifying model transformations, e.g.~\cite{narayanan.karsai.towardsverifyingmt,hulsbusch.konig.rensink.modeltransformation}, using many different techniques~\cite{amrani.overview,survey.mt}. In earlier work, we have developed a formal verification technique to determine whether the definition of a model transformation preserves specific safety or liveness properties, \emph{regardless} of the model it is applied on~\cite{engelen.wijs.modevva,wijs.facs13,wijs.engelen.tacas,wijs.engelen.refiner}. It is applicable on any modelling language with a formal semantics that can be captured by \emph{Labelled Transition Systems} (\LTSs), i.e.\ it must be action (or event) based. For example, in our research we focus on the visual modelling language SLCO~\cite{engelen.thesis}, which allows the specification and development of concurrent and distributed systems by defining sets of (interacting) finite state machines.

In our setting, the semantics of models is captured by \LTSs, and the semantics of transformations is captured by systems of pairs of \LTSs, describing which patterns in an input \LTS should be transformed into which new patterns. By applying the technique from~\cite{engelen.wijs.modevva,wijs.facs13,wijs.engelen.tacas,wijs.engelen.refiner} on those pairs of \LTSs, we are able to determine whether transformation is guaranteed to preserve the structure of any \LTS w.r.t.\ a particular temporal logic formula expressing a desired functional property. If that is the case, then models for which this property holds can be safely transformed.

For example, consider the \LTS $C_0$ on the left in Figure~\ref{fig:example}. It describes a system that can alternatingly receive messages $m$ and perform computations. The transformation rule $r_0$ in the middle defines that after each receiving of a message, a postprocessing step should be added. The result of applying the rule on the \LTS on the left in the figure, produces the \LTS $T(C_0)$ on the right. More interesting cases involve multiple \LTSs describing the semantics of different components running concurrently, and interaction between those components is taken into account.

This setting allows to formally reason about model transformations, which has a number of practical applications. For instance, if one desires to develop a system running on specific hardware, then an abstract model can be transformed to make it compatible with that hardware. In~\cite{wijs.facs13}, a transformation rule system is given to transform multi-party communication in models, i.e.\ involving more than two parties at once, into a number of two-party communications following a specific protocol. Multi-party communication can be useful in modelling languages to reason about system behaviour at an abstract level, but if the eventual implementation cannot use it, then such a transformation rule system is useful to automatically remove it at some point in the development process. The ability to verify the transformation rule system in isolation means that we are able to determine that it will always produce a correct output model when applied on a correct input model. Other elaborate examples of \LTS transformations are given in~\cite{wijs.engelen.tacas,wijs.engelen.refiner}.

\begin{figure}[t]
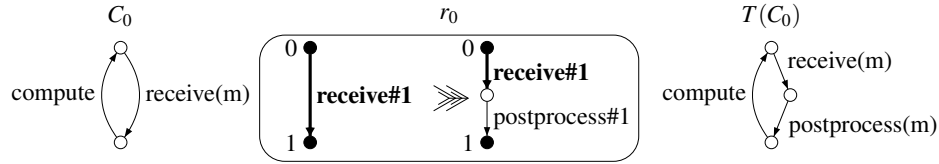

\centering
\scalebox{0.84}{
\begin{gpicture}(100,100)(0,0)
	\node[Nw=2,Nh=2,linecolor=white](c0)(20,25){$C_0$}
	\node[Nw=2,Nh=2](s0)(20,20){}
	\node[Nw=2,Nh=2](s1)(20,5){}
	\drawedge[curvedepth=3](s0,s1){{receive(m)}}
	\drawedge[curvedepth=3](s1,s0){{compute}}

	\node[Nw=60,Nh=20](box)(72,12){}

	\node[Nw=2,Nh=2,linecolor=white](r0)(72,25){$r_0$}
	\node[Nw=2,Nh=2,linecolor=white](t1)(47,20){0}
	\node[Nw=2,Nh=2,fillcolor=black](p00)(50,20){}
	\node[Nw=2,Nh=2,fillcolor=black](p01)(50,5){}
	\node[Nw=2,Nh=2,linecolor=white](t1)(47,5){1}
	\drawedge[linewidth=0.5,AHangle=30](p00,p01){{\bf receive\#1}}

	\drawline[AHnb=2,AHangle=30,AHLength=4,AHlength=0,AHdist=2,ATnb=2,ATangle=150,
					ATLength=2,ATlength=0,ATdist=2](72,12)(75,12)

	\node[Nw=2,Nh=2,linecolor=white](t1)(75,20){0}
	\node[Nw=2,Nh=2,fillcolor=black](p10)(78,20){}
	\node[Nw=2,Nh=2](p11)(78,12.5){}
	\node[Nw=2,Nh=2,fillcolor=black](p12)(78,5){}
	\node[Nw=2,Nh=2,linecolor=white](t1)(75,5){1}
	\drawedge[linewidth=0.5,AHangle=30](p10,p11){{\bf receive\#1}}
	\drawedge(p11,p12){{postprocess\#1}}

	\node[Nw=2,Nh=2,linecolor=white](c0p)(123,25){$T(C_0)$}
	\node[Nw=2,Nh=2](t0)(123,20){}
	\node[Nw=2,Nh=2](t1)(123,5){}
	\node[Nw=2,Nh=2](t2)(126,12.5){}
	\drawedge[ELpos=75](t0,t2){{receive(m)}}
	\drawedge[ELpos=5](t2,t1){{postprocess(m)}}
	\drawedge[curvedepth=3](t1,t0){{compute}}
	
\end{gpicture}
}
\caption{Example of an LTS transformation}
\label{fig:example}
\end{figure}

For model transformation, it is crucial that the transformation is always \emph{terminating} and \emph{confluent}, i.e.\ that transformation is guaranteed to finish, and that it always leads to the same solution, i.e.\ reduct, independent of the order in which matches are processed. This is important, since a user defining how a particular system should be transformed typically has a specific resulting model in mind. Therefore, if a rule system is not confluent, it usually means that the user made some mistake. Both termination and confluence of transformation systems has been studied before; for instance, in~\cite{termination}, criteria are given to determine whether a system is terminating or not.

Confluence has been the subject of research in for example~\cite{lambers.ehrig.orejas.2005,lambers.ehrig.orejas.2008,plump.first,plump.confluence05,plump.confluence2010}. From term rewriting, we know that a rewrite system is confluent if it is both terminating and locally confluent~\cite{huet.confluence}. A system is locally confluent iff all possible conflicts between two transformation applications, i.e.\ direct transformations, can be resolved according to the Critical Pair Lemma~\cite{plump.first,plump.confluence05,plump.confluence2010}. For a terminating system, determining confluence therefore boils down to doing two operations: firstly, to construct so-called \emph{critical pairs} representing possible conflicts between two direct transformations, and secondly, to try to resolve all constructed critical pairs.

In earlier work on critical pair detection, general graphs have always been considered, meaning that vertices and edges may or may not have labels, edges may or may not be directed, and graphs can consist of several disconnected subgraphs. Such a general approach is of course very useful, but when considering a more restricted setting, it may be possible to detect critical pairs more efficiently. The complexity of standard critical pair detection for general graph structures is exponential in the number of vertices and edges in the left patterns of two transformation rules, since all possible overlappings between the two left patterns need to be taken into account. In~\cite{lambers.ehrig.orejas.2005}, it is demonstrated that if one transformation rule deletes elements and the other does not, a check with a linear complexity can be obtained, and when both rules do not delete elements, a check with quadratic complexity is possible.

The contributions of this paper follow from the observation that in our setting, graphs are directed, edge-labelled graphs, i.e.\ \LTSs, with all vertices connected with each other via edges (ignoring the direction). These structural properties can be exploited further, along the reasoning of~\cite{lambers.ehrig.orejas.2005}, to find critical pairs more efficiently; the presence of edge labels allows to check in constant time in some cases, and furthermore, we are also able to define a check for critical pairs with quadratic complexity for the case that two rules both delete elements. The main conclusion for our setting is that transitions, as opposed to states, turn out to play a crucial role in the detection of critical pairs.

\paragraph*{Contribution}
In this paper, we propose a new critical pair detection algorithm when working with \LTSs as opposed to general graphs. When formally reasoning about model transformations as transformations of \LTSs, we can immediately benefit from the new algorithm since it can handle many cases more efficiently than existing detection techniques~\cite{lambers.ehrig.orejas.2005,lambers.ehrig.orejas.2008,plump.first,plump.confluence05,plump.confluence2010}. It uses a novel approach, constructing partial morphisms between \LTSs. Whenever such a partial morphism meets certain requirements, a conflict can be directly derived from it. Although the worst-case complexity of the algorithm is comparable to that of algorithms proposed in related work, it is very efficient in particular cases. The circumstances of those cases are explained in detail. Besides that, we also propose an algorithm to try to resolve detected conflicts, which is based on observations made in~\cite{plump.confluence2010}, but we focus on our particular setting.

\paragraph*{Roadmap}

Section~\ref{sec:prelim} presents the basic notions, in particular LTS and LTS transformation. In Section~\ref{sec:confluence}, we investigate conflict detection. From this, we construct a conflict detection algorithm in Section~\ref{sec:alg} and a conflict resolution algorithm in the same section. Finally, Section~\ref{sec:conclusions} contains our conclusions and pointers for future work.

\section{Background}
\label{sec:prelim}

In this paper, we focus on action-based semantics of (concurrent) systems. Such semantics are often captured using \emph{Labelled Transition Systems} (\LTSs), indicating how a system as a whole or an individual component in a system can change state by performing particular actions.

\begin{definition}[Labelled Transition System]
An LTS~\graph is a tuple $\langle \states{\graph},$ $\actions{\graph},$ $\transitions{\graph} \rangle$, where \states{\graph} is a (finite) set of states, $\actions{\graph}$ is a set of actions, and $\transitions{\graph} \subseteq \states{\graph} \times \actions{\graph} \times \states{\graph}$ is a transition relation.
Actions in \actions{\graph} are denoted by $a$, $b$, $c$, etc.
We use $s_1 \xrightarrow{a}_\graph s_2$ to denote $\langle s_1, a, s_2 \rangle \in \transitions{\graph}$.
If $s_1 \xrightarrow{a}_\graph s_2$, this means that in~\graph, an action $a$ can be performed in state~$s_1$, leading to state~$s_2$.
\end{definition}



We use operations on \LTSs such as intersection and difference in the usual graph-theoretical way.
Finally, an \LTS is \emph{weakly connected} iff the undirected version of an \LTS is a single connected component (from each state, there is path to each other state). 

\begin{figure}[t]
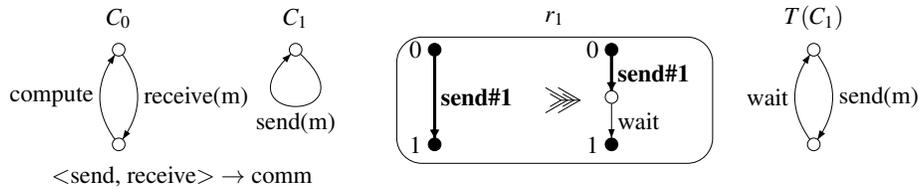

\centering
\scalebox{0.84}{
\begin{gpicture}(100,100)(0,0)
	\node[Nw=2,Nh=2,linecolor=white](c1)(23,25){$C_1$}
	\node [Nw=2,Nh=2](t0)(23,20){}
	\drawloop[loopangle=270](t0){{send(m)}}

	\node[Nw=2,Nh=2,linecolor=white](c0)(-5,25){$C_0$}
	\node[Nw=2,Nh=2](s0)(-5,20){}
	\node[Nw=2,Nh=2](s1)(-5,5){}
	\drawedge[curvedepth=3](s0,s1){{receive(m)}}
	\drawedge[curvedepth=3](s1,s0){{compute}}
	\node[Nw=5,Nh=3,Nmr=0,linecolor=white](tl1)(5,0){$<$send, receive$>$ $\to$ comm}

	\node[Nw=50,Nh=20](box)(64,12){}

	\node[Nw=2,Nh=2,linecolor=white](c0)(64,25){$r_1$}

	\node[Nw=2,Nh=2,linecolor=white](t1)(42,20){0}
	\node[Nw=2,Nh=2,fillcolor=black](p00)(45,20){}
	\node[Nw=2,Nh=2,fillcolor=black](p01)(45,5){}
	\node[Nw=2,Nh=2,linecolor=white](t1)(42,5){1}
	\drawedge[linewidth=0.5,AHangle=30](p00,p01){{\bf send\#1}}

	\drawline[AHnb=2,AHangle=30,AHLength=4,AHlength=0,AHdist=2,ATnb=2,ATangle=150,
					ATLength=2,ATlength=0,ATdist=2](65,12)(68,12)

	\node[Nw=2,Nh=2,linecolor=white](t1)(70,20){0}
	\node[Nw=2,Nh=2,fillcolor=black](p10)(73,20){}
	\node[Nw=2,Nh=2](p11)(73,12.5){}
	\node[Nw=2,Nh=2,fillcolor=black](p12)(73,5){}
	\node[Nw=2,Nh=2,linecolor=white](t1)(70,5){1}
	\drawedge[linewidth=0.5,AHangle=30](p10,p11){{\bf send\#1}}
	\drawedge(p11,p12){{wait}}

	\node[Nw=2,Nh=2,linecolor=white](c1p)(105,25){$T(C_1)$}
	\node[Nw=2,Nh=2](tt0)(105,20){}
	\node[Nw=2,Nh=2](tt1)(105,5){}
	\drawedge[curvedepth=3](tt0,tt1){{send(m)}}
	\drawedge[curvedepth=3](tt1,tt0){{wait}}

\end{gpicture}
}
\caption{Transforming networks of \LTSs}
\label{fig:transnet}
\end{figure}

In the context of systems consisting of a finite number of concurrent components, we actually represent system semantics as \emph{networks of \LTSs}~\cite{lang05}, where the potential behaviour of each component is described by a separate \LTS, and a synchronisation mechanism is defined describing the potential for those \LTSs to interact. For example, consider the network on the left in Figure~\ref{fig:transnet}, which besides \LTS $C_0$ from Figure~\ref{fig:example} also contains the potential behaviour of a component $C_1$ that can at any time send a message. Below the \LTSs, a synchronisation rule is defined stating that {\it send} and {\it receive} actions can synchronise, leading to a {\it comm} action in the \LTS of the system. For this to be possible, the parameters of {\it send} and {\it receive} must be identical, i.e.\ they must involve the same message $m$.

When considering transformation rule systems applicable on networks of \LTSs, confluence depends on whether the rules in the rule system do not give rise to conflicts in either of the individual \LTSs, so it again boils down to considering single \LTSs. For this reason, we do not consider networks of \LTSs in most of this paper.

\paragraph*{Transformation}

In our setting, changes applied on a concurrent system model are represented by \LTS \emph{transformation rules} applied on the semantics of the components of the model, i.e.\ on their \LTSs. We only consider weakly connected \LTSs as the semantics of components, since naturally, the possible states that components can be in should be reachable from their initial state. To reason about the changes, we define the notions of a rule, and matches of rules on component \LTSs. But first, we introduce the notion of \emph{\LTS morphisms}.

\begin{definition}[\LTS morphism]
\label{def:ltsmorphism}
An \emph{\LTS morphism} $f : \graph_0 \to \graph_1$ between two \LTSs $\graph_0 = \langle \states{\graph_0}$, $\actions{\graph_0}$, $\transitions{\graph_0} \rangle$, $\graph_1 = \ltstuple{\graph_1}$ is a pair of functions $f = \langle f_\states{}: \states{\graph_0} \to \states{\graph_1}, f_\transitions{}: \transitions{\graph_0} \to \transitions{\graph_1} \rangle$ which preserve sources, targets, and transition labels, i.e.\ for all $s \xrightarrow{a}_{\graph_0} s'$, we have $f_\transitions{}(s \xrightarrow{a}_{\graph_0} s') = f_\states{}(s) \xrightarrow{a}_{\graph_1} f_\states{}(s')$.
\end{definition}

We denote the existence of an \emph{injective} LTS morphism $f$ from an \LTS $\graph_0$ to an \LTS $\graph_1$, meaning that $f_{\states{}}$ and $f_{\transitions{}}$ are injective, by $\graph_0 \sqsubseteq \graph_1$, and say that $\graph_0$ and $\graph_1$ are \emph{isomorphic}, denoted by $\graph_0 \simeq \graph_1$, iff there exists a morphism $f: \graph_0 \to \graph_1$ such that both $f_\states{}$ and $f_\transitions{}$ are bijections. An \emph{\LTS inclusion} $i : \graph_0 \to \graph_1$ is an \LTS morphism with for all $s \in \states{\graph_0}$, $i_\states{}(s) = s$, and for all $s \xrightarrow{a}_{\graph_0} s'$, $i_\transitions{}(s \xrightarrow{a}_{\graph_0} s') = s \xrightarrow{a}_{\graph_1} s'$. \LTS $\graph_0$ is a sub-\LTS of $\graph_1$, denoted by $\graph_0 \subseteq \graph_1$, iff there exists an \LTS inclusion $i : \graph_0 \to \graph_1$.
Finally, we denote the fact that a morphism is undefined for a particular state or transition with $\perp$, for example $f_{\states{}}(s) = \perp$ means that $f_{\states{}}(s)$ is undefined.

\begin{definition}[Transformation Rule]
\label{def:lts-transformation:transformationrule}
A \emph{transformation rule} $r = \langle \Llts \xleftarrow{f} \Klts \xrightarrow{g} \Rlts \rangle$ consists of two \LTS morphisms $f:\Klts \to \Llts$, $g: \Klts \to \Rlts$, where $\Klts \to \Llts$ is an inclusion. \LTSs $\Llts$ and $\Rlts$ are both weakly connected, and are called the \emph{left} and \emph{right patterns} of $r$. \LTS $\Klts$ is the \emph{interface}.
\end{definition}

We consider injective transformation rules, meaning that $\Klts \to \Rlts$ is injective. With $\states{\Llts \setminus \Klts}$, we refer to the states $s$ in \Llts\ that are not represented in \Klts, i.e.\ $f_{\states{}}^{-1}(s) = \perp$. A similar convention is used for the functions $f_{\transitions{}}$, $g_{\states{}}$, and $g_{\transitions{}}$. States $s \in \states{\Llts}$ for which $f_{\states{}}^{-1}(s)$ is defined are called \emph{glue-states}. 

\begin{definition}[Rule Match]
\label{def:lts-transformation:transmatch}
A transformation rule $r = \langle \Llts \xleftarrow{f} \Klts \xrightarrow{g} \Rlts \rangle$ has a \emph{match} $m: \Llts \to \graph$ on an \LTS~$\graph = \ltstuple{\graph}$ iff $m: \Llts \to \graph$ is an injective \LTS morphism and $\forall s \in \states{\Llts \setminus \Klts}, p \in \states{\graph}:$
	\begin{itemize}[noitemsep]
 	\item $m_\states{}(s) \xrightarrow{a}_\graph p \implies \exists s' \in \states{\Llts}. s\xrightarrow{a}_{\Llts} s' \wedge m_\states{}(s') = p$;
 	\item $p \xrightarrow{a}_\graph m_\states{}(s) \implies \exists s' \in \states{\Llts}. s'\xrightarrow{a}_{\Llts} s \wedge m_\states{}(s') = p$.
	\end{itemize}
\end{definition}

The conditions in Def.~\ref{def:lts-transformation:transmatch} correspond with the \emph{gluing conditions} of the \emph{double-pushout} (DPO) method \cite{gt.book} for graph transformation, preventing so-called \emph{dangling transitions}, which are transitions where only the source or target state will be removed, but not both.
It expresses that for a state $s$ to be removed, all connected transitions must be removed as well. 

Let $\graph$, $\graphresult$ be \LTSs, and $m: \Llts \to \graph$ a match for rule $r$. Then $\graph$ \emph{directly transforms} $\graphresult$ by $r$ and $m$, denoted by $\graph \Rightarrow_{r,m} \graphresult$, iff there are two pushouts as in Figure~\ref{fig:DPO}.


\begin{figure}[t]
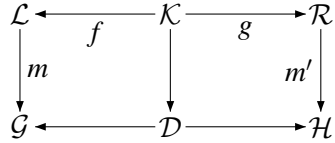

\centering
\begin{gpicture}(100,100)(0,0)
	\node[Nw=4,Nh=4,Nmr=0,linecolor=white](lr)(20,20){$\Klts$}
	\node[Nw=4,Nh=4,Nmr=0,linecolor=white](l)(0,20){$\Llts$}
	\node[Nw=4,Nh=4,Nmr=0,linecolor=white](r)(40,20){$\Rlts$}
	\node[Nw=4,Nh=4,Nmr=0,linecolor=white](g)(0,5){$\graph$}
	\node[Nw=4,Nh=4,Nmr=0,linecolor=white](d)(20,5){$\mathcal{D}$}
	\node[Nw=4,Nh=4,Nmr=0,linecolor=white](h)(40,5){$\graphresult$}
	
	\drawedge(lr,l){$f$}
	\drawedge[ELside=r](lr,r){$g$}
	\drawedge(l,g){$m$}
	\drawedge(d,g){}
	\drawedge(d,h){}
	\drawedge[ELside=r](r,h){$m'$}
	\drawedge(lr,d){}

\end{gpicture}
\caption{Double-pushout diagram}
\label{fig:DPO}
\end{figure}

Direct transformation is defined as follows, with $m':\Rlts \to \graphresult$ a match between the right pattern and the result of the transformation.

\begin{definition}[Direct Transformation]
\label{def:lts-transformation:transformation}
The \emph{direct transformation} $\graph \Rightarrow_{r,m} \graphresult$ of an LTS~$\graph = \ltstuple{\graph}$ according to a rule~$r$ and a given match $m:\Llts \to \graph$ is defined as
$\graphresult = \langle \states{\graphresult}, \actions{\graphresult}, \transitions{\graphresult} \rangle$, where
\begin{itemize}[noitemsep]
\item
$\states{\graphresult} =
( \states{\graph}
  \setminus
  \{m_\states{}(s) \mid s \in (\states{\Llts \setminus \Klts}) \}
)
\cup
(\states{\Rlts \setminus \Klts})$;

\item
$\transitions{\graphresult} =
(
  \transitions{\graph}
  \setminus
  \{ m_\transitions{}(\langle s, a, s'\rangle) \mid s\xrightarrow{a}_{\Llts \setminus \Klts} s' \}
)
\cup
\{ m_\transitions{}'(\langle s, a, s' \rangle) \mid s\xrightarrow{a}_{\Rlts \setminus \Klts} s' \}
$;
\item
$\actions{\graphresult} = \{ a \mid \exists \langle s, a, s' \rangle \in \transitions{\graphresult} \}$.
\end{itemize}
\end{definition}

The new set of states \states{\graphresult} consists of \states{\graph} without the states that correspond to the states in the left pattern that are not represented in the interface, i.e.\ the removed states, and with new representatives, here represented in the match $m'$, of the states in the right pattern that are not represented in the interface, i.e.\ the newly added states.
In a similar way, $\transitions{\graphresult}$ consists of the transitions in $\transitions{\graph}$ without the transitions corresponding to left pattern transitions that are not represented in the interface, and with transitions corresponding to right pattern transitions that are not represented in the interface.


An example of a transformation rule introducing a new action is given in the middle of Figure~\ref{fig:example}. Black states with the same index correspond with each other, i.e.\ for two such states $s \in \states{\Llts}$, $t \in \states{\Rlts}$, we have $g_{\states{}}(f_{\states{}}^{-1}(s)) = t$. This also holds for the highlighted transition labelled \emph{receive(m)}.

It is crucial to note at this point that we expect transformation rules to be \emph{well-specified}, i.e.\ that they specify that input is actually altered, and not replaced by something that can be considered equivalent. In particular, we assume that a transition $s \xrightarrow{a}_\graph s$ is not replaced by a new transition between states $s$ and $s'$ with the same label $a$, nor that a state $s$ is replaced by another state $\hat s$ without transforming any of the transitions connected to $s$ (note that \LTSs to be transformed are weakly connected, so $s$ always has connected transitions).

Sets of rules together make up a \emph{rule system} $\Sigma$. Transformation of an \LTSs $\graph$ according to a rule system $\Sigma$ involves identifying all possible matches for each $r \in \Sigma$ on $\graph$, and applying transformation on those matches. A \emph{transformation} from $\graph$ to $\graphresult$ is a sequence of direct transformations $\graph = \graph_0 \Rightarrow \ldots \Rightarrow \graph_n = \graphresult$, with $n \geq 0$. We denote this by $\graph \Rightarrow^*_\Sigma \graphresult$. 

Figure~\ref{fig:transnet} presents a transformation rule $r_1$ which can only be applied on $C_1$, leading to \LTS $T(C_1)$. When combining $r_1$ with $r_0$ into a rule system $\Sigma$, it is clear that $\Sigma$ is confluent w.r.t.\ the network given in Figure~\ref{fig:transnet}, since $r_0$ and $r_1$ are not applicable on the same \LTSs. The question that remains is whether $\Sigma$ is confluent for arbitrary single \LTSs as well.



\section{Conflicts Between Direct Transformations}
\label{sec:confluence}

By Newman's Lemma~\cite{newman}, a terminating transformation system is confluent iff it is \emph{locally confluent}, i.e.\ if for all direct transformations $\graphresult_0 \Leftarrow_{\rz, m_0} \graph \Rightarrow_{\ro, m_1} \graphresult_1$, there is a common reduct $\graphresult$ with $\graphresult_0 \Rightarrow^*_\Sigma \graphresult$ and $\graphresult_1 \Rightarrow^*_\Sigma \graphresult$. To determine local confluence, first of all, it has been shown~\cite{plump.first} that if two direct transformations are \emph{parallel independent}, then they are locally confluent. In the following, we reason about two \LTS transformation rules $\rz = \langle \Llts^\rz, \Rlts^\rz \rangle$ and $\ro = \langle \Llts^\ro, \Rlts^\ro \rangle$.


\begin{definition}[Parallel Independence]
\label{def:parind}
Direct transformations $\graphresult_0 \Leftarrow_{\rz,m_0} \graph$ $\Rightarrow_{\ro, m_1} \graphresult_1$ are \emph{parallel independent} iff
\[
m_0(\Llts^\rz) \cap m_1(\Llts^\ro) \subseteq m_0(\Klts^\rz) \cap m_1(\Klts^\ro)
\]
\end{definition}


The intuition behind Def.~\ref{def:parind} is that if two matches of one or two rules on an \LTS only overlap w.r.t.\ the interfaces of those two rules, then the related direct transformations are parallel independent since applying one direct transformation does not invalidate the match for the other direct transformation.
We say that two direct transformations are \emph{in conflict} iff they are not parallel independent. The presence of a conflict can cause the transformation system to be not locally confluent (specifically, if the two derivations do not lead to a common reduct~\cite{gt.book}). Informally, the conflict is caused because $\rz$ deletes something that $\ro$ uses and/or $\ro$ deletes something that $\rz$ uses. A concrete conflict can be represented by a \emph{critical pair}, which defines an \LTS on which two matches of the given pair of rules exist that imply derivations that are in conflict. Clearly, in such an \LTS, the two matches must overlap. 

\begin{definition}[Critical Pair]
\label{def:critpair}
Direct transformations $\graphresult_0 \Leftarrow_{\rz, m_0} \graph \Rightarrow_{\ro, m_1} \graphresult_1$ form a \emph{critical pair} iff they are not parallel independent and $\graph = m_0(\Llts^\rz) \cup m_1(\Llts^\ro)$.
\end{definition}

Furthermore, we require that $m_0 \neq m_1$ if $\rz = \ro$, and we equate isomorphic critical pairs, which informally means that two critical pairs are different if either there exists no isomorphism between their $\graph$'s, or their matches $m_0$, $m_1$ are different.

The main task when detecting conflicts is to construct a suitable \emph{conflict situation} $\graph$ for pairs of rules $\rz$, $\ro$ that gives rise to a conflict. Such a $\graph$ should be minimal, in the sense that there does not exist an \LTS $\graph'$ with $\graph' \subset \graph$ and matches $m_0':\Llts^\rz \to \graph'$, $m_1': \Llts^\ro \to \graph'$ such that $\graphresult_0' \Leftarrow_{\rz, m_0'} \graph' \Rightarrow_{\ro, m_1'} \graphresult_1'$ is also a critical pair.



%

In this section, we focus on how to construct a suitable $\graph$ efficiently. First, we establish that in order to have a conflict in $\graph$, $m_0(\Llts^\rz)$ and $m_1(\Llts^\ro)$ must at least overlap in one transition. This relies on the fact that our \LTSs $\graph$ are weakly connected. If a rule specifies that all states matched on a left pattern state $s$ should be removed, then so must all transitions that connect with those states. By the fact that such transitions always exist ($\graph$ is weakly connected) and Def.~\ref{def:lts-transformation:transmatch}, it follows that the rule must also specify explicitly that these transitions must be removed. Hence, a conflict between rules concerning the removal of states also must involve the removal of transitions.



\begin{lemma}
\label{lem:parind}
Direct transformations $\graphresult_0 \Leftarrow_{\rz,m_0} \graph \Rightarrow_{\ro, m_1} \graphresult_1$ are parallel independent iff
\[
\transitions{m_0(\Llts^\rz) \cap m_1(\Llts^\ro)} \subseteq \transitions{m_0(\Klts^\rz) \cap m_1(\Klts^\ro)}
\]
\end{lemma}
\begin{proof}
The \emph{if} case is trivial. If the \LTS $m_0(\Llts^\rz) \cap m_1(\Llts^\ro)$ is contained in the \LTS $m_0(\Klts^\rz) \cap m_1(\Klts^\ro)$, then all transitions of $m_0(\Llts^\rz) \cap m_1(\Llts^\ro)$ are contained in $m_0(\Klts^\rz) \cap m_1(\Klts^\ro)$.

For the \emph{only if} case, we reason towards a contradiction. Assume that the direct transformations are not parallel independent, i.e.\ $m_0(\Llts^\rz) \cap m_1(\Llts^\ro) \not\subseteq m_0(\Klts^\rz) \cap m_1(\Klts^\ro)$, but that $\transitions{m_0(\Llts^\rz) \cap m_1(\Llts^\ro)} \subseteq \transitions{m_0(\Klts^\rz) \cap m_1(\Klts^\ro)}$.
Then, we must have that $\states{m_0(\Llts^\rz) \cap m_1(\Llts^\ro)} \not\subseteq \states{m_0(\Klts^\rz) \cap m_1(\Klts^\ro)}$. We will prove that this cannot be the case by reasoning towards a contradiction. Let $p \in \states{m_0(\Llts^\rz) \cap m_1(\Llts^\ro)}$ and $p \not\in \states{m_0(\Klts^\rz) \cap m_1(\Klts^\ro)}$. Then, there must exist $s \in \states{\Llts^\rz}$ with $m_{0,\states{}}(s) = p$ and $t \in \states{\Llts^\ro}$ with $m_{1,\states{}}(t) = p$. Since $p \not\in \states{m_0(\Klts^\rz) \cap m_1(\Klts^\ro)}$, we have $s \not\in \states{\Klts^\rz}$ and $t \not\in \states{\Klts^\ro}$. Because $\graph$ is weakly connected and $\ro$ is well-specified, $p$ must have at least one in- or outgoing transition which will be removed by the direct transformation $\graph \Rightarrow_{\ro, m_1} \graphresult_1$. Let us assume that this is an incoming transition $\hat p \xrightarrow{a}_{\graph} p$ (the case of an outgoing transition is similar). Since $m_1(t) = p$ and $t \in \states{\Llts^\ro \setminus \Klts^\ro}$, by Def.~\ref{def:lts-transformation:transmatch}, there must be a transition $\hat t \xrightarrow{a}_{\Llts^\ro} t$, with $m_{1,\states{}}(\hat t) = \hat p$, and $m_{1,\transitions{}}(\hat t \xrightarrow{a}_{\Llts^\ro} t) = \hat p \xrightarrow{a}_{\graph} p$. Similarly, there must be an $\hat s \in \states{\Llts^\rz}$ with $\hat s \xrightarrow{a}_{\Llts^\rz} s$, $m_{0,\states{}}(\hat s) = \hat p$ and $m_{0,\transitions{}}(\hat s \xrightarrow{a}_{\Llts^\rz} s) = \hat p \xrightarrow{a}_{\graph} p$. This means that both $\hat p \xrightarrow{a}_{m_0(\Llts^\rz)} p$ and $\hat p \xrightarrow{a}_{m_1(\Llts^\ro)} p$. Also, since $\hat p \xrightarrow{a}_{\graph} p$ is set for removal, we must have $\langle \hat p, a, p \rangle \not\in \transitions{m_0(\Klts^\rz)}$ and $\langle \hat p, a, p \rangle \not\in \transitions{m_1(\Klts^\ro)}$. But then, $\transitions{m_0(\Llts^\rz) \cap m_1(\Llts^\ro)} \not\subseteq \transitions{m_0(\Klts^\rz) \cap m_1(\Klts^\ro)}$, and we have a contradiction.
\end{proof}

The following lemma expresses that parallel independence of rules $\rz$, $\ro$ can be concluded if the sets of transition labels of $\rz$ and $\ro$ satisfy certain conditions. We use $\actions{\Llts \mid \Klts}$ to refer to the labels for which there exists at least one transition in $\Llts$ that is not represented in $\Klts$, i.e.\ there exists an $s \xrightarrow{a}_\Llts s'$ for which $f_{\transitions{}}^{-1}(s \xrightarrow{a}_\Llts s') = \perp$.

\begin{figure}[t]
\centering
\scalebox{0.84}{
\begin{gpicture}(100,100)(0,0)

	\node[Nw=34,Nh=20](box)(55,12){}

	\node[Nw=2,Nh=2,linecolor=white](c0)(55,25){$\rz$}

	\node[Nw=2,Nh=2,linecolor=white](t1)(42,20){0}
	\node[Nw=2,Nh=2,fillcolor=black](p00)(45,20){}
	\node[Nw=2,Nh=2,fillcolor=black](p01)(45,12.5){}
	\node[Nw=2,Nh=2,linecolor=white](t1)(42,12.5){1}
	\node[Nw=2,Nh=2,fillcolor=black](p02)(45,5){}
	\node[Nw=2,Nh=2,linecolor=white](t2)(42,5){2}
	\drawedge[linewidth=0.5,AHangle=30](p00,p01){{\bf a}}
	\drawedge(p01,p02){b}

	\drawline[AHnb=2,AHangle=30,AHLength=4,AHlength=0,AHdist=2,ATnb=2,ATangle=150,
					ATLength=2,ATlength=0,ATdist=2](55,12)(58,12)

	\node[Nw=2,Nh=2,linecolor=white](t1)(62,20){0}
	\node[Nw=2,Nh=2,fillcolor=black](p10)(65,20){}
	\node[Nw=2,Nh=2,fillcolor=black](p11)(65,12.5){}
	\node[Nw=2,Nh=2,linecolor=white](t1)(62,12.5){1}
	\node[Nw=2,Nh=2,fillcolor=black](p12)(65,5){}
	\node[Nw=2,Nh=2,linecolor=white](t2)(62,5){2}
	\drawedge[linewidth=0.5,AHangle=30](p10,p11){{\bf a}}
	\drawedge(p11,p12){{c}}

	\node[Nw=34,Nh=20](box)(95,12){}

	\node[Nw=2,Nh=2,linecolor=white](c0)(95,25){$\ro$}

	\node[Nw=2,Nh=2,linecolor=white](t1)(82,20){0}
	\node[Nw=2,Nh=2,fillcolor=black](p00)(85,20){}
	\node[Nw=2,Nh=2,fillcolor=black](p01)(85,12.5){}
	\node[Nw=2,Nh=2,linecolor=white](t1)(82,12.5){1}
	\node[Nw=2,Nh=2,fillcolor=black](p02)(85,5){}
	\node[Nw=2,Nh=2,linecolor=white](t2)(82,5){2}
	\drawedge[linewidth=0.5,AHangle=30](p00,p01){{\bf a}}
	\drawedge(p01,p02){d}
	\drawedge[curvedepth=-6](p02,p00){a}

	\drawline[AHnb=2,AHangle=30,AHLength=4,AHlength=0,AHdist=2,ATnb=2,ATangle=150,
					ATLength=2,ATlength=0,ATdist=2](95,12)(98,12)

	\node[Nw=2,Nh=2,linecolor=white](t1)(102,20){0}
	\node[Nw=2,Nh=2,fillcolor=black](p10)(105,20){}
	\node[Nw=2,Nh=2,fillcolor=black](p11)(105,12.5){}
	\node[Nw=2,Nh=2,linecolor=white](t1)(102,12.5){1}
	\node[Nw=2,Nh=2,fillcolor=black](p12)(105,5){}
	\node[Nw=2,Nh=2,linecolor=white](t2)(102,5){2}
	\drawedge[linewidth=0.5,AHangle=30](p10,p11){{\bf a}}
	\drawedge(p11,p12){{e}}
	
	\node[Nw=2,Nh=2,linecolor=white](g)(145,25){$\graph$}

	\node[Nw=2,Nh=2](m00)(145,20){}
	\node[Nw=2,Nh=2](m01)(145,12.5){}
	\node[Nw=2,Nh=2](m02)(145,5){}
	\node[Nw=2,Nh=2](m03)(155,20){}
	\drawedge(m00,m01){a}
	\drawedge(m01,m02){d}
	\drawedge[curvedepth=-6](m02,m00){a}
	\drawedge(m00,m03){b}

	\node[Nw=2,Nh=2,linecolor=white](t1)(137,15){$\Leftarrow$}
	\node[Nw=2,Nh=2,linecolor=white](t1)(137,17.5){$\rz$}

	\node[Nw=2,Nh=2](m00)(120,20){}
	\node[Nw=2,Nh=2](m01)(120,12.5){}
	\node[Nw=2,Nh=2](m02)(120,5){}
	\node[Nw=2,Nh=2](m03)(130,20){}
	\drawedge(m00,m01){a}
	\drawedge(m01,m02){d}
	\drawedge[curvedepth=-6](m02,m00){a}
	\drawedge(m00,m03){c}
	
	\node[Nw=2,Nh=2,linecolor=white](t1)(160,15){$\Rightarrow$}
	\node[Nw=2,Nh=2,linecolor=white](t1)(160,17.5){$\ro$}

	\node[Nw=2,Nh=2](m00)(167,20){}
	\node[Nw=2,Nh=2](m01)(167,12.5){}
	\node[Nw=2,Nh=2](m02)(167,5){}
	\node[Nw=2,Nh=2](m03)(177,20){}
	\drawedge(m00,m01){a}
	\drawedge(m01,m02){e}
	\drawedge(m00,m03){b}
\end{gpicture}
}
\caption{Two rules $\rz$ and $\ro$ with a possible conflict situation $\graph$}
\label{fig:conflictexample}
\end{figure}

\begin{lemma}
\label{lem:actionsets}
Direct transformations $\graphresult_0 \Leftarrow_{\rz,m_0} \graph \Rightarrow_{\ro, m_1} \graphresult_1$ are parallel independent if $\actions{\Llts^\rz \mid \Klts^\rz} \cap \actions{\Llts^\ro} = \emptyset$ and $\actions{\Llts^\ro \mid \Klts^\ro} \cap \actions{\Llts^\rz} = \emptyset$.
\end{lemma}
\begin{proof}
By reasoning towards a contradiction. Assume that $(\actions{\Llts^\rz \mid \Klts^\rz}) \cap \actions{\Llts^\ro} = \emptyset$ and $(\actions{\Llts^\ro \mid \Klts^\ro}) \cap \actions{\Llts^\rz} = \emptyset$, but that the direct transformations are not parallel independent. From Lemma~\ref{lem:parind}, it follows that $\transitions{m_0(\Llts^\rz) \cap m_1(\Llts^\ro)}$ must be non-empty. So, there must exist a transition $s \xrightarrow{a}_{m_0(\Llts^\rz) \cap m_1(\Llts^\ro)} s'$ that is not in $m_0(\Klts^\rz) \cap m_1(\Klts^\ro)$. From this, it follows that $a \in \actions{\Llts^\rz}$, $a \in \actions{\Llts^\ro}$ and $a \not\in \actions{\Klts^\rz}$. But then, $\actions{\Llts^\rz \mid \Klts^\rz} \cap \actions{\Llts^\ro} \neq \emptyset$, and we have a contradiction.
\end{proof}

Lemma~\ref{lem:actionsets} will be used as a first check in a conflict detection algorithm in the next section. If for two rules, the mentioned intersection of action sets of left patterns is empty, then it is not possible to construct critical pairs. Since this can be checked in linear time, assuming that set membership can be checked in constant time, it helps to avoid more involved conflict detection for many cases in practice.

Consider the example illustrated in Figure~\ref{fig:conflictexample}. For the two rules $\rz$, $\ro$, we have $\actions{\Llts^\rz \mid \Klts^\rz} = \{ b \}$ and $\actions{\Llts^\ro} = \{ a, d \}$, hence $\actions{\Llts^\rz \mid \Klts^\rz} \cap \actions{\Llts^\ro} = \emptyset$, but $\actions{\Llts^\ro \mid \Klts^\ro} = \{ a, d\}$ and $\actions{\Llts^\rz} = \{ a, b \}$, so $\actions{\Llts^\ro \mid \Klts^\ro} \cap \actions{\Llts^\rz} = \{ a \}$. This means that there is potential for a conflict situation, and a valid conflict situation is actually illustrated on the right in Figure~\ref{fig:conflictexample}. In the given \LTS $\graph$, applying the direct transformation defined by $\Llts^\ro$ matched on the lower part of $\graph$ results in an \LTS on which $\Llts^\rz$ can still be matched. However, note that $\Llts^\rz$ can be matched on $\graph$ involving the curved $a$-transition and the $b$-transition. Since $\rz$ removes the matched $a$-transition, this means that the possible match of $\Llts^\ro$ on $\graph$ is removed when applying the direct transformation of $\rz$.


Next, we concentrate on constructing minimal conflict situations $\graph$ for pairs of rules $\rz$, $\ro$. We will do so by constructing a relation between states $s \in \states{\Llts^\rz}$, $t \in \states{\Llts^\ro}$ that expresses the potential to match them on the same state $p$ in an arbitrary \LTS. If a non-empty relation can be constructed, then a conflict situation can be derived from it.

Two states $s \in \states{\Llts^\rz}$, $t \in \states{\Llts^\ro}$ can only be matched on the same state $p$ if their in- and outgoing transitions are in some sense compatible. Since $\Llts^\rz$ and $\Llts^\ro$ are weakly connected, we know that $s$ and $t$ must have in- and/or outgoing transitions. To reason about these, we define the notion of a \emph{context \LTS} of a state.

\begin{definition}[Context \LTS]
\label{def:contextlts}
Given an \LTS $\graph$, we say that for a state $p \in \states{\graph}$, the \emph{context \LTS} $\cLTS{\graph}(p) = \langle \states{p}, \actions{p}, \transitions{p} \rangle$ is defined as
follows:
\begin{itemize}[noitemsep]
\item $\states{p} = \{ p' \mid \exists a \in \actions{\graph}. p\xrightarrow{a}_\graph p' \vee p'\xrightarrow{a}_\graph p \} \cup \{ p \}$;
\item $\actions{p} = \{ a \mid \exists p' \in \states{\graph}. p\xrightarrow{a}_\graph p' \vee p'\xrightarrow{a}_\graph p \}$;
\item $\transitions{p} = \{ \langle p, a, p' \rangle \mid p\xrightarrow{a}_\graph p' \} \cup \{ \langle p', a, p \rangle \mid p'\xrightarrow{a}_\graph p \}$.
\end{itemize}
\end{definition}

\begin{figure}[t]
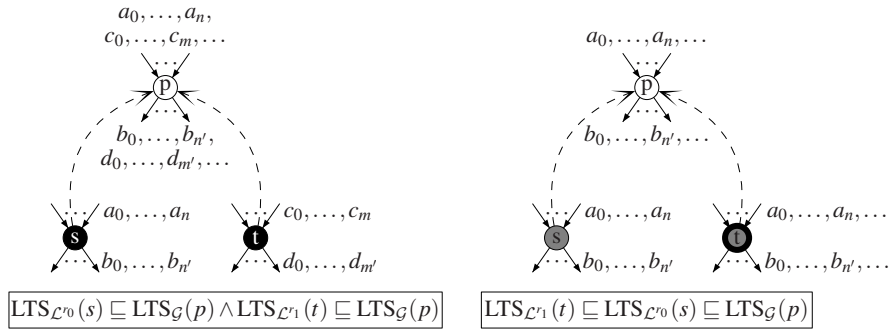

\centering
\scalebox{0.8}{
\begin{gpicture}(119,100)(0,0)
	\node[Nw=4.0,Nh=4.0,NLangle=0.0,fillcolor=black](gs1)(10,10){\textcolor{white}{s}}
	\imark[iangle=125,ilength=5](gs1)
	\imark[iangle=55,ilength=5](gs1)
	\node[Nw=2,Nh=2,Nmr=0,linecolor=white](tl1)(10.3,14){$\cdots$}
	\node[Nw=5,Nh=3,Nmr=0,linecolor=white](tl1)(22,14){$a_0,\ldots,a_n$}
	\fmark[fangle=305,ilength=5](gs1)
	\node[Nw=2,Nh=2,Nmr=0,linecolor=white](tl1)(10.3,6){$\cdots$}
	\fmark[fangle=235,ilength=5](gs1)
	\node[Nw=5,Nh=3,Nmr=0,linecolor=white](tl1)(22,6){$b_0,\ldots,b_{n'}$}

	\node[Nw=4.0,Nh=4.0,NLangle=0.0,fillcolor=black](gt1)(40,10){\textcolor{white}{t}}
	\imark[iangle=125,ilength=5](gt1)
	\imark[iangle=55,ilength=5](gt1)
	\node[Nw=2,Nh=2,Nmr=0,linecolor=white](tl1)(40.3,14){$\cdots$}
	\node[Nw=5,Nh=3,Nmr=0,linecolor=white](tl1)(52,14){$c_0,\ldots,c_m$}
	\fmark[fangle=305,ilength=5](gt1)
	\node[Nw=2,Nh=2,Nmr=0,linecolor=white](tl1)(40.3,6){$\cdots$}
	\fmark[fangle=235,ilength=5](gt1)
	\node[Nw=5,Nh=3,Nmr=0,linecolor=white](tl1)(52.5,6){$d_0,\ldots,d_{m'}$}

	\node[Nw=4.0,Nh=4.0,NLangle=0.0](gp1)(25,35){p}
	\imark[iangle=125,ilength=5](gp1)
	\imark[iangle=55,ilength=5](gp1)
	\node[Nw=2,Nh=2,Nmr=0,linecolor=white](tl1)(25.3,39){$\cdots$}
	\node[Nw=5,Nh=3,Nmr=0,linecolor=white](tl1)(25,47){$a_0,\ldots,a_n,$}
	\node[Nw=5,Nh=3,Nmr=0,linecolor=white](tl1)(25,43){$c_0,\ldots, c_m, \ldots$}
	\fmark[fangle=305,ilength=5](gp1)
	\node[Nw=2,Nh=2,Nmr=0,linecolor=white](tl1)(25.3,31){$\cdots$}
	\fmark[fangle=235,ilength=5](gp1)
	\node[Nw=5,Nh=3,Nmr=0,linecolor=white](tl1)(25,27){$b_0,\ldots,b_{n'},$}
	\node[Nw=5,Nh=3,Nmr=0,linecolor=white](tl1)(25,23){$d_0, \ldots, d_{m'}, \ldots$}
	
	\drawedge[curvedepth=7,dash={1.5}0,AHlength=1,AHLength=4](gs1,gp1){}
	\drawedge[curvedepth=-7,dash={1.5}0,AHlength=1,AHLength=4](gt1,gp1){}

	\node[Nw=72,Nh=6,Nmr=0](tl1)(35,-2){$\cLTS{\Llts^\rz}(s) \sqsubseteq \cLTS{\graph}(p) \wedge \cLTS{\Llts^\ro}(t) \sqsubseteq \cLTS{\graph}(p)$}
	
	\node[Nw=4.0,Nh=4.0,NLangle=0.0,fillcolor=gray](ns1)(90,10){s}
	\imark[iangle=125,ilength=5](ns1)
	\imark[iangle=55,ilength=5](ns1)
	\node[Nw=2,Nh=2,Nmr=0,linecolor=white](tl1)(90.3,14){$\cdots$}
	\node[Nw=5,Nh=3,Nmr=0,linecolor=white](tl1)(102,14){$a_0,\ldots,a_n$}
	\fmark[fangle=305,ilength=5](ns1)
	\node[Nw=2,Nh=2,Nmr=0,linecolor=white](tl1)(90.3,6){$\cdots$}
	\fmark[fangle=235,ilength=5](ns1)
	\node[Nw=5,Nh=3,Nmr=0,linecolor=white](tl1)(102,6){$b_0,\ldots,b_{n'}$}

	\node[Nw=4.0,Nh=4.0,NLangle=0.0,fillcolor=gray,linewidth=1](nt1)(120,10){t}
	\imark[iangle=125,ilength=5](nt1)
	\imark[iangle=55,ilength=5](nt1)
	\node[Nw=2,Nh=2,Nmr=0,linecolor=white](tl1)(120.3,14){$\cdots$}
	\node[Nw=5,Nh=3,Nmr=0,linecolor=white](tl1)(135,14){$a_0,\ldots,a_n,\ldots$}
	\fmark[fangle=305,ilength=5](nt1)
	\node[Nw=2,Nh=2,Nmr=0,linecolor=white](tl1)(120.3,6){$\cdots$}
	\fmark[fangle=235,ilength=5](nt1)
	\node[Nw=5,Nh=3,Nmr=0,linecolor=white](tl1)(135,6){$b_0,\ldots,b_{n'},\ldots$}

	\node[Nw=4.0,Nh=4.0,NLangle=0.0](np1)(105,35){p}
	\imark[iangle=125,ilength=5](np1)
	\imark[iangle=55,ilength=5](np1)
	\node[Nw=2,Nh=2,Nmr=0,linecolor=white](tl1)(105.3,39){$\cdots$}
	\node[Nw=5,Nh=3,Nmr=0,linecolor=white](tl1)(105,43){$a_0,\ldots, a_n, \ldots$}
	\fmark[fangle=305,ilength=5](np1)
	\node[Nw=2,Nh=2,Nmr=0,linecolor=white](tl1)(105.3,31){$\cdots$}
	\fmark[fangle=235,ilength=5](np1)
	\node[Nw=5,Nh=3,Nmr=0,linecolor=white](tl1)(105,27){$b_0,\ldots,b_{n'},\ldots$}
	
	\drawedge[curvedepth=7,dash={1.5}0,AHlength=1,AHLength=4](ns1,np1){}
	\drawedge[curvedepth=-7,dash={1.5}0,AHlength=1,AHLength=4](nt1,np1){}

	\node[Nw=55,Nh=6,Nmr=0](tl1)(105,-2){$\cLTS{\Llts^\ro}(t) \sqsubseteq \cLTS{\Llts^\rz}(s) \sqsubseteq \cLTS{\graph}(p)$}

\end{gpicture}
}
\caption{Matching rule left pattern states on the same \LTS states}
\label{fig:mapping}
\end{figure}

Figure~\ref{fig:mapping} shows the conditions under which we are able to match two left pattern states on the same state.
On the left, the case where both states are glue is covered. In the figure, glue-states are coloured black.
The figure expresses that any two glue-states $s$ and $t$ can be matched on a state $p$ as long as $p$ has matchable transitions for all the transitions of $s$ and $t$. Say that state $s$ has $n+1$ incoming  transitions with labels $a_0, \ldots, a_n$, and $n'+1$ outgoing transitions with labels $b_0, \ldots, b_{n'}$, and state $t$ has $m+1$ incoming transitions with labels $c_0, \ldots, c_m$, and $m'+1$ outgoing transitions with labels $d_0, \ldots, d_{m'}$, then $p$ should have matchable transitions with all those labels. Of course, some transitions of $s$ and $t$ may be matched on common transitions of $p$, i.e.\ the two matches together could be non-injective.

Below the figure on the left, this condition is formalised as follows: we must have that $\cLTS{\Llts^\rz}(s) \sqsubseteq \cLTS{\graph}(p) \wedge \cLTS{\Llts^\ro}(t) \sqsubseteq \cLTS{\graph}(p)$. This directly follows from the fact that matches are injective \LTS morphisms (Def.~\ref{def:lts-transformation:transmatch}). 

On the right in Figure~\ref{fig:mapping}, the condition for the possibility to match two states on the same state is given for the case that at least one of those states is non-glue. In the figure, non-glue states are coloured grey, and state $t$ may be either non-glue or glue. The condition expresses that for all incoming and outgoing transitions of $t$, $s$ must have corresponding transitions with the same label. This is formalised as $\cLTS{\Llts^\ro}(t) \sqsubseteq \cLTS{\Llts^\rz}(s) \sqsubseteq \cLTS{\graph}(p)$. The idea is that if $t$ has transitions that $s$ does not have, then matching $s$ and $t$ on the same state $p$ would not be possible due to the gluing conditions. Again, the fact that $\cLTS{\Llts^\rz}(s) \sqsubseteq \cLTS{\graph}(p)$ and $\cLTS{\Llts^\ro}(t) \sqsubseteq \cLTS{\graph}(p)$ should hold follows from the fact that matches are injective \LTS morphisms. That $\cLTS{\Llts^\ro}(t) \sqsubseteq \cLTS{\Llts^\rz}(s)$ needs to hold follows from the following lemma.

\begin{lemma}
Let $s \in \states{\Llts^\rz \setminus \Klts^\rz}$, $t \in \states{\Llts^\ro}$ be states. Then there can be no matches $m_0: \Llts^\rz \to \graph$, $m_1: \Llts^\ro \to \graph$ on an arbitrary \LTS $\graph$ with $m_{0,\states{}}(s) = p$ and $m_{1,\states{}}(t) = p$ for some $p \in \states{\graph}$ if $\cLTS{\Llts^\ro}(t) \not\sqsubseteq \cLTS{\Llts^\rz}(s)$.
\end{lemma}
\begin{proof}
By reasoning towards a contradiction.
Assume that $\cLTS{\Llts^\ro}(t) \not\sqsubseteq \cLTS{\Llts^\rz}(s)$ holds, and that we have matches $m_0: \Llts^\rz \to \graph$, $m_1: \Llts^\ro \to \graph$ with $m_{0,\states{}}(s) = p$ and $m_{1,\states{}}(t) = p$. Since $\cLTS{\Llts^\ro}(t) \not\sqsubseteq \cLTS{\Llts^\rz}(s)$ and $\Llts^\rz$ and $\Llts^\ro$ are weakly connected, we must have that at least one transition in $\cLTS{\Llts^\ro}(t)$ cannot be mapped on a transition in $\cLTS{\Llts^\rz}(s)$. Say that this is an incoming transition of $t$. The case that it is an outgoing transition of $t$ is similar. We refer to this transition as $t' \xrightarrow{a}_{\Llts^\ro} t$. Since, by Def.~\ref{def:contextlts}, $\cLTS{\Llts^\ro}(t) \subseteq \Llts^\ro$, and since $m_1$ is an injective \LTS morphism, $m_{1,\states{}}(t')$ and $m_{1,\transitions{}}(t' \xrightarrow{a}_{\Llts^\ro} t)$ must be defined. Let us say that $m_{1,\states{}}(t') = p'$, and $m_{1,\transitions{}}(t' \xrightarrow{a}_{\Llts^\ro} t) = p' \xrightarrow{a}_{\graph} p$. By the fact that $\cLTS{\Llts^\rz}(s)$ contains all the incoming and outgoing transitions of $s$ (Def.~\ref{def:contextlts}), and by the facts that $t'\xrightarrow{a}_{\Llts^\ro} t$ cannot be mapped on a transition in $\cLTS{\Llts^\rz}(s)$ and $\cLTS{\Llts^\rz}(s) \subseteq \Llts^\rz$, it follows that $p'$ and $p' \xrightarrow{a}_\graph p$ are not matched by the \LTS morphism $m_0$. But, since $s \in \states{\Llts^\rz \setminus \Klts^\rz}$ and $p'\xrightarrow{a}_{\graph} m_{0,\states{}}(s)$, by Def.~\ref{def:lts-transformation:transmatch}, we must have that there exists an $s' \in \states{\Llts^\rz}$ such that $m _{0,\states{}}(s') = p'$, and we have a contradiction.
\end{proof}

Note that the condition on the right in Figure~\ref{fig:mapping} implies what needs to hold if both $s$ and $t$ are non-glue. If $s$ is non-glue, we must have that $\cLTS{\Llts^\ro}(t) \sqsubseteq \cLTS{\Llts^\rz}(s) \sqsubseteq \cLTS{\graph}(p)$, but if $t$ is non-glue, we must have $\cLTS{\Llts^\rz}(s) \sqsubseteq \cLTS{\Llts^\ro}(t) \sqsubseteq \cLTS{\graph}(p)$, i.e.\ we must have that $\cLTS{\Llts^\rz}(s) \simeq \cLTS{\Llts^\ro}(t)$.


Figure~\ref{fig:mapping} gives rise to defining a relation between left patterns of rules, where states $s$ and $t$ are related iff it is conceivable to construct a situation (in the form of an \LTS) in which matches $m_0$ and $m_1$ relate $s$ and $t$ to a common state $p$, and likewise for transitions. Having such a relation, it follows from Lemma~\ref{lem:parind} that if it at least relates two transitions of which at least one is not represented in the interface of the corresponding rule, then the inferred situation is a conflict situation, i.e.\ there are direct transformations that are in conflict. We will use such a relation later on to reason about all possible conflicts involving two given rules, by iterating over all pairs of states from their left patterns.

Next, we define this relation between left patterns, and after that, we explain how a conflict situation can be constructed from a concrete relation.

\begin{definition}[Conflict Compatibility Morphism]
\label{def:conflict}
Let $s \in \states{\Llts^\rz}$, $t \in \states{\Llts^\ro}$. A partial \LTS morphism $f: \Llts^\rz \to \Llts^\ro$ is a \emph{conflict compatibility morphism} if it is injective and $f_{\states{}}(s) = t$ implies that
\begin{itemize}[noitemsep]
\item If $s \in \states{\Llts^\rz \setminus \Klts^\rz}$ and $t \in \states{\Llts^\ro}$	then
	\begin{itemize}[noitemsep]
		\item if $t \xrightarrow{a}_{\Llts^\ro} t'$ then $s \xrightarrow{a}_{\Llts^\rz} s'$ with $f_{\states{}}(s') = t'$, $f_{\transitions{}}(s \xrightarrow{a}_{\Llts^\rz} s') = t \xrightarrow{a}_{\Llts^\ro} t'$;
		\item if $t \xleftarrow{a}_{\Llts^\ro} t'$ then $s \xleftarrow{a}_{\Llts^\rz} s'$ with $f_{\states{}}(s') = t'$, $f_{\transitions{}}(s \xleftarrow{a}_{\Llts^\rz} s') = t \xleftarrow{a}_{\Llts^\ro} t'$.
	\end{itemize}
\item If $s \in \states{\Llts^\rz}$ and $t \in \states{\Llts^\ro \setminus \Klts^\ro}$	then
	\begin{itemize}[noitemsep]
		\item if $s \xrightarrow{a}_{\Llts^\rz} s'$ then $t \xrightarrow{a}_{\Llts^\ro} t'$ with $f_{\states{}}(s') = t'$, $f_{\transitions{}}(s \xrightarrow{a}_{\Llts^\rz} s') = t \xrightarrow{a}_{\Llts^\ro} t'$;
		\item if $s \xleftarrow{a}_{\Llts^\rz} s'$ then $t \xleftarrow{a}_{\Llts^\ro} t'$ with $f_{\states{}}(s') = t'$, $f_{\transitions{}}(s \xleftarrow{a}_{\Llts^\rz} s') = t \xleftarrow{a}_{\Llts^\ro} t'$.
	\end{itemize}
\item If $s \in \states{\Llts^\rz \setminus \Klts^\rz}$ and $t \in \states{\Llts^\ro \setminus \Klts^\ro}$	then
	\begin{itemize}[noitemsep]
		\item if $s \xrightarrow{a}_{\Llts^\rz} s'$ then $t \xrightarrow{a}_{\Llts^\ro} t'$ with $f_{\states{}}(s') = t'$, $f_{\transitions{}}(s \xrightarrow{a}_{\Llts^\rz} s') = t \xrightarrow{a}_{\Llts^\ro} t'$;
		\item if $s \xleftarrow{a}_{\Llts^\rz} s'$ and $t \xleftarrow{a}_{\Llts^\ro} t'$, then $f_{\states{}}(s') = t'$, $f_{\transitions{}}(s \xleftarrow{a}_{\Llts^\rz} s') = t \xleftarrow{a}_{\Llts^\ro} t'$.
		\item if $t \xrightarrow{a}_{\Llts^\ro} t'$ then $s \xrightarrow{a}_{\Llts^\rz} s'$ with $f_{\states{}}(s') = t'$, $f_{\transitions{}}(s \xrightarrow{a}_{\Llts^\rz} s') = t \xrightarrow{a}_{\Llts^\ro} t'$;
		\item if $s \xleftarrow{a}_{\Llts^\rz} s'$ then $t \xleftarrow{a}_{\Llts^\ro} t'$ with $f_{\states{}}(s') = t'$, $f_{\transitions{}}(s \xleftarrow{a}_{\Llts^\rz} s') = t \xleftarrow{a}_{\Llts^\ro} t'$.
	\end{itemize}
\end{itemize}
\end{definition}

Note that Def.~\ref{def:conflict} does not state anything about the case that both $s$ and $t$ are glue-states. Unlike in the other cases, in which the gluing conditions are relevant because at least one non-glue state is involved, two glue-states can always be related to each other. This means that an \LTS morphism which only relates glue-states and no transitions is also a conflict compatibility morphism. However, by Lemma~\ref{lem:parind}, such a morphism does not directly represent a conflict situation. We are not interested in just any conflict compatibility morphism, but those for which $f_{\transitions{}}$ is defined for some transitions. In fact, given two left patterns $\Llts^\rz$, $\Llts^\ro$ and two states $s \in \states{\Llts^\rz}$, $t \in \states{\Llts^\ro}$, we are interested in the \emph{largest} conflict compatibility morphism $f$ for which $f_{\states{}}(s) = t$, and its \emph{domain of definition}, i.e.\ the part of $\Llts^\rz$ for which $f$ is defined, is a \emph{weakly connected} \LTS. 

\begin{figure}[t]
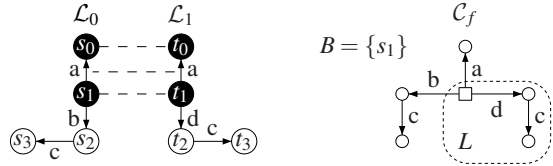

\centering
\scalebox{0.84}{
\begin{gpicture}(100,100)(0,0)
	
	\node[Nw=2,Nh=2,linecolor=white](g)(0,25){$\Llts_0$}

	\node[Nw=4,Nh=4,fillcolor=black](s0)(0,20){\textcolor{white}{$s_0$}}
	\node[Nw=4,Nh=4,fillcolor=black](s1)(0,12.5){\textcolor{white}{$s_1$}}
	\node[Nw=4,Nh=4](s2)(0,5){$s_2$}
	\node[Nw=4,Nh=4](s3)(-10,5){$s_3$}
	\drawedge(s1,s0){a}
	\drawedge[ELside=r](s1,s2){b}
	\drawedge(s2,s3){c}

	\node[Nw=2,Nh=2,linecolor=white](g)(15,25){$\Llts_1$}

	\node[Nw=4,Nh=4,fillcolor=black](t0)(15,20){\textcolor{white}{$t_0$}}
	\node[Nw=4,Nh=4,fillcolor=black](t1)(15,12.5){\textcolor{white}{$t_1$}}
	\node[Nw=4,Nh=4](t2)(15,5){$t_2$}
	\node[Nw=4,Nh=4](t3)(25,5){$t_3$}
	\drawedge[ELside=r](t1,t0){a}
	\drawedge(t1,t2){d}
	\drawedge(t2,t3){c}

	\drawline[AHnb=0,dash={1.5}0](1,16)(14,16)
	\drawedge[AHnb=0,dash={1.5}0](s0,t0){}
	\drawedge[AHnb=0,dash={1.5}0](s1,t1){}

	\node[Nw=2,Nh=2,linecolor=white](g)(44,20){$B = \{ s_1 \}$}

	\node[Nw=2,Nh=2,linecolor=white](g)(60,25){$\conflict_f$}

	\node[Nw=2,Nh=2](t0)(60,20){}
	\node[Nw=2,Nh=2,Nmr=0](t1)(60,12.5){}
	\node[Nw=2,Nh=2](t2)(70,12.5){}
	\node[Nw=2,Nh=2](t3)(70,5){}
	\node[Nw=2,Nh=2](t4)(50,12.5){}
	\node[Nw=2,Nh=2](t5)(50,5){}
	\drawedge[ELside=r](t1,t0){a}
	\drawedge[ELside=r](t1,t2){d}
	\drawedge(t2,t3){c}
	\drawedge[ELside=r](t1,t4){b}
	\drawedge(t4,t5){c}

	\node[Nw=2,Nh=2,linecolor=white](g)(60,5){$L$}
	\node[Nw=17,Nh=13,dash={0.5}0](box)(65,8){}

\end{gpicture}
}
\caption{A conflict compatibility morphism between left patterns $\Llts_0$, $\Llts_1$, and the corresponding conflict situation $\conflict_f$}
\label{fig:conflictcompexample}
\end{figure}

For example, consider the two \LTSs $\Llts_0$ and $\Llts_1$ on the left in Figure~\ref{fig:conflictcompexample}. A conflict compatibility morphism $f$ with $f_{\states{}}(s_1) = t_1$ could be defined without relating any other states and transitions, but it would not represent a conflict. Instead, the largest possible morphism $f$ with $f_{\states{}}(s_1) = t_1$ and a weakly connected domain of definition also relates $s_0$ with $t_0$, and the $a$-transitions. In particular, $s_2 \xrightarrow{c} s_3$ and $t_2 \xrightarrow{c} t_3$ are not related by $f$, since that would make its domain of definition not weakly connected.

Note that for two states $s$, $t$, there can be more than one conflict compatibility morphism of interest, particularly if there are multiple options to relate states and transitions.

In the remainder of this paper, each conflict compatibility morphism is the largest possible with a weakly connected domain of definition, in the sense that no morphism can be constructed that contains it and also has a weakly connected domain of definition.

Given $s \in \states{\Llts^\rz}$ and $t \in \states{\Llts^\ro}$, we can now construct conflict compatibility morphisms $f$. From $f$, $\Llts^\rz$, and $\Llts^\ro$, we can construct a conflict situation. For this, we use $m_0$, $m_1$ to map $\Llts^\rz$ and $\Llts^\ro$ to isomorphic \LTS structures.

\begin{definition}[Conflict Situation]
\label{def:conflictsituation}
Let $f : \Llts^\rz \to \Llts^\ro$ be a conflict compatibility morphism, and $m_0(\Llts^\rz)$, $m_1(\Llts^\ro)$ \LTSs isomorphic to $\Llts^\rz$ and $\Llts^\ro$, respectively. Then, a \emph{conflict situation} \LTS $\conflict_f$ can be constructed as follows: first, determine the boundary $B$ of $f$ consisting of all states $s \in \Llts^\rz$ such that $f_{\states{}}(s)$ is defined, but for some transition $s \xrightarrow{a}_{\Llts^\rz} s'$ (or $s \xleftarrow{a}_{\Llts^\rz} s'$), $f_{\transitions{}}(s \xrightarrow{a}_{\Llts^\rz} s')$ (or $f_{\transitions{}}(s \xleftarrow{a}_{\Llts^\rz} s')$) is not. Then, $L = m_1(\Llts^\ro \setminus f(\Llts^\rz)) \cup m_1(f(B))$ can be glued to $m_0(\Llts^\rz)$ by merging each state $s \in m_0(B)$ with the corresponding state $s' \in m_1(f(B))$. The result is $\conflict_f$.
\end{definition}

On the right of Figure~\ref{fig:conflictcompexample}, the conflict situation is presented which results from applying Def.~\ref{def:conflictsituation} on the conflict compatibility morphism between $\Llts_0$ and $\Llts_1$ given on the left in the figure. The boundary $B$ is defined as $B=\{ s_1 \}$, since $f_\states{}(s_1) = t_1$, but $f_\transitions{}(s_1 \xrightarrow{b}_{\Llts_0} s_2) = \perp$. Then, $L$ is the \LTS isomorphic to $\Llts_1$ without $t_0$ and the $a$-transition between $t_0$ and $t_1$ (indicated in $C_f$ in the figure), and $L$ is glued to an \LTS isomorphic to $\Llts_0$ by merging the states related to $s_1$ and $t_1$ (resulting in the square state in the figure).

\section{Conflict Detection and Resolution Algorithms}
\label{sec:alg}

\begin{algorithm}
\caption{Conflict detection algorithm}
\label{alg:conflict}
{\scriptsize
\begin{algorithmic}[2]
\REQUIRE Rules $\rz = \langle \Llts^\rz, \Rlts^\rz \rangle$, $\ro = \langle \Llts^\ro, \Rlts^\ro \rangle$
\ENSURE Returns set of conflicts between $\rz$ and $\ro$
\STATE $C = \emptyset$
\IF{$\actions{\Llts^\rz \setminus \Klts^\rz} \cap \actions{\Llts^\ro} = \emptyset$ {\bf and} $\actions{\Llts^\ro \setminus \Klts^\ro} \cap \actions{\Llts^\rz} = \emptyset$}
	\RETURN $\emptyset$ \COMMENT{Lemma~\ref{lem:actionsets}}
\ENDIF
\IF{$\Llts^\rz \simeq \Klts^\rz \wedge \Llts^\ro \simeq \Klts^\ro$}
	\RETURN $\emptyset$ \COMMENT{See~\cite{lambers.ehrig.orejas.2005}}
\ENDIF
\FORALL{$s \in \states{\Llts^\rz}$, $t \in \states{\Llts^\ro}$}
	\IF {$\actions{\it out}(s) \cap \vectornot{\actions{\it out}}(t) \neq \emptyset$ {\bf or} $\vectornot{\actions{\it out}}(s) \cap \actions{\it out}(t) \neq \emptyset$}
		\FORALL {conflict compatibility morphisms $f: \Llts^\rz \to \Llts^\ro$ with $f_{\states{}}(s) = t$}
			\IF {$f_{\transitions{}}$ is defined for at least one transition}
			\STATE add $\conflict_f$ to $C$ \COMMENT{Definition~\ref{def:conflict}}
			\ENDIF
		\ENDFOR
	\ENDIF
\ENDFOR
\RETURN $C$
\end{algorithmic}
}
\end{algorithm}

The findings presented in Section~\ref{sec:confluence} can be used to construct a new conflict detection algorithm. It is presented in Alg.~\ref{alg:conflict}. Given two rules, the algorithm tries to determine whether there can be an \LTS for which it is possible to construct direct transformations that are in conflict. The decision procedures are sorted by their complexity. If full analysis is needed, i.e.\ all possible conflict compatibility morphisms have to be computed, then attempts can be restricted to those pairs of states $s \in \states{\Llts^\rz}$, $t \in \states{\Llts^\ro}$ that share an outgoing transition label, and for at least one of the two states, an outgoing transition with that label will be removed when transforming. This directly follows from Lemma~\ref{lem:parind}, which implies that in order to have a conflict, at least one transition needs to be involved which is matched on by both rules $\rz$ and $\ro$ and removed by at least one of these rules. To formalise this, we use the following notation: $\actions{\it out}(s) = \{ a \in \actions{\Llts^\rz} \mid \exists s \xrightarrow{a}_{\Llts^\rz} s'\}$ is the set of labels of outgoing transitions of $s$, and $\vectornot{\actions{\it out}}(s) = \{ a \in \actions{\Llts^\rz} \mid \exists s \xrightarrow{a}_{\Llts^\rz} s' . f^{-1}_{\transitions{}}(s \xrightarrow{a}_{\Llts^\rz} s') = \perp\}$ is the set of labels of outgoing transitions that are set for removal. At line 2, the action sets are compared, which can be done in $\bigo(|\actions{}|)$ time, with $\actions{} = \actions{\Llts^\rz} \cup \actions{\Llts^\ro}$. At line 4, we use a check from~\cite{lambers.ehrig.orejas.2005}, based on the fact that two non-deleting rules can never be in conflict. This can be determined in $\bigo(|\states{}|+|\transitions{}|)$ time, with $|\states{}|$ and $|\transitions{}|$ the total number of states and transitions in the two left rule patterns together. Next, full checking for conflict compatibility morphisms (lines 7-10) requires worst-case to compare the two left \LTS patterns for all pairs of states, i.e.\ its complexity is $\bigo(|\states{}|^2 \cdot |\transitions{}| \cdot \log |\states{}|)$, since a comparison of two \LTSs can be done in $\bigo(|\transitions{}| \cdot \log |\states{}|)$ time, using the equivalence checking algorithm of Paige \& Tarjan~\cite{paige.tarjan}.

Compared to earlier work, our detection algorithm has a number of advantages. First of all, comparison of the action sets can be done in linear time, and is, unlike other special case optimisations, such as those in~\cite{lambers.ehrig.orejas.2005}, also applicable when both $\rz$ and $\ro$ remove some transitions. Second of all, not all possible pairs of states $s \in \states{\Llts^\rz}$, $t \in \states{\Llts^\ro}$ need to be considered in detail. Just by considering their outgoing transitions first can we quickly resolve many combinations in practice. This exploits the fact that \LTSs are weakly connected, or more specifically, that most states have outgoing transitions.

\begin{figure}[t]
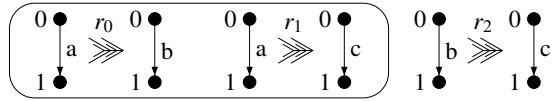

\centering
\scalebox{0.84}{
\begin{gpicture}(100,100)(0,0)

	\node[Nw=60,Nh=15](box)(72,15){}

	\node[Nw=2,Nh=2,linecolor=white](t1)(47,20){0}
	\node[Nw=2,Nh=2,fillcolor=black](p00)(50,20){}
	\node[Nw=2,Nh=2,fillcolor=black](p01)(50,10){}
	\node[Nw=2,Nh=2,linecolor=white](t1)(47,10){1}
	\drawedge(p00,p01){a}

	\drawline[AHnb=2,AHangle=30,AHLength=4,AHlength=0,AHdist=2,ATnb=2,ATangle=150,
					ATLength=2,ATlength=0,ATdist=2](57,15)(60,15)
	\node[Nw=2,Nh=2,linecolor=white](t1)(57,19){$\rz$}

	\node[Nw=2,Nh=2,linecolor=white](t1)(62,20){0}
	\node[Nw=2,Nh=2,fillcolor=black](p10)(65,20){}
	\node[Nw=2,Nh=2,fillcolor=black](p12)(65,10){}
	\node[Nw=2,Nh=2,linecolor=white](t1)(62,10){1}
	\drawedge(p10,p12){b}

	\node[Nw=2,Nh=2,linecolor=white](t1)(77,20){0}
	\node[Nw=2,Nh=2,fillcolor=black](p00)(80,20){}
	\node[Nw=2,Nh=2,fillcolor=black](p01)(80,10){}
	\node[Nw=2,Nh=2,linecolor=white](t1)(77,10){1}
	\drawedge(p00,p01){a}

	\drawline[AHnb=2,AHangle=30,AHLength=4,AHlength=0,AHdist=2,ATnb=2,ATangle=150,
					ATLength=2,ATlength=0,ATdist=2](87,15)(90,15)
	\node[Nw=2,Nh=2,linecolor=white](t1)(87,19){$\ro$}

	\node[Nw=2,Nh=2,linecolor=white](t1)(92,20){0}
	\node[Nw=2,Nh=2,fillcolor=black](p10)(95,20){}
	\node[Nw=2,Nh=2,fillcolor=black](p12)(95,10){}
	\node[Nw=2,Nh=2,linecolor=white](t1)(92,10){1}
	\drawedge(p10,p12){c}

	\node[Nw=2,Nh=2,linecolor=white](t1)(107,20){0}
	\node[Nw=2,Nh=2,fillcolor=black](p00)(110,20){}
	\node[Nw=2,Nh=2,fillcolor=black](p01)(110,10){}
	\node[Nw=2,Nh=2,linecolor=white](t1)(107,10){1}
	\drawedge(p00,p01){b}

	\drawline[AHnb=2,AHangle=30,AHLength=4,AHlength=0,AHdist=2,ATnb=2,ATangle=150,
					ATLength=2,ATlength=0,ATdist=2](117,15)(120,15)
	\node[Nw=2,Nh=2,linecolor=white](t1)(117,19){$r_2$}

	\node[Nw=2,Nh=2,linecolor=white](t1)(122,20){0}
	\node[Nw=2,Nh=2,fillcolor=black](p10)(125,20){}
	\node[Nw=2,Nh=2,fillcolor=black](p12)(125,10){}
	\node[Nw=2,Nh=2,linecolor=white](t1)(122,10){1}
	\drawedge(p10,p12){c}
\end{gpicture}
}
\caption{A Critical Pair may be resolvable}
\label{fig:critresolve}
\end{figure}

It is a known fact in graph transformation that the existence of critical pairs does not guarantee that a transformation system is not confluent. For example, consider the system in Figure~\ref{fig:critresolve}. Rules $\rz$ and $\ro$ are clearly in conflict, since they both concern an $a$-transition in their left pattern. They also define different transformation results, namely a $b$- and a $c$-transition, respectively, so direct transformations on an \LTS $\graph$ consisting of a transition $s \xrightarrow{a}_{\graph} s'$ constitute a critical pair. However, consider that there is a third rule $r_2$ in the system, which transforms $b$-transitions into $c$-transitions. Then $\graph$ can be transformed to an \LTS consisting of a single $c$-transition either by first applying $\rz$ and then $r_2$, or by applying $\ro$. The conflict represented by the critical pair can be resolved.

Another example is the conflict situation in Figure~\ref{fig:conflictexample}. It cannot be resolved if the rule system only consists of rules $r_0$ and $r_1$. Applying first the direct transformation of $r_1$ removes the match of $m_0$ on $\graph$, and applying first the direct transformation of $r_0$, followed by the one of $r_1$ leads to a different \LTS, in which instead of the $b$-transition there is now a $c$-transition.

Since the existence of critical pairs does not mean that a transformation system is not confluent, a necessary condition needs to be found for a critical pair to actually be an example why a system is not confluent. Plump~\cite{plump.confluence2010} demonstrates that for so-called \emph{coverable} transformation systems of hypergraphs, i.e.\ graphs where the edges can be associated with multiple source and target vertices, it suffices to show that all the critical pairs are \emph{strongly joinable}, meaning that independent of which of the two involved direct transformations is applied on the conflict situation, the system can transform the resulting graph to a graph that is structurally equivalent to the graph that can be obtained if the other direct transformation had been applied first. Since \LTSs are a special kind of hypergraph, and since our \LTSs are always coverable (an important criterium is that \LTSs can be extended with a cover consisting of transitions with fresh labels), we can directly take the result from~\cite{plump.confluence2010} for our setting. To define strong joinability formally, we first need to define the notion of a \emph{track morphism}, based on~\cite{plump.confluence05}. Such a morphism explicitly involves relations between states based on the fact that they have been matched by the same interface state. 

\begin{definition}[Track Morphism]
Given a direct transformation $\graph \Rightarrow_{r,m} \graphresult$, the \emph{track morphism} $\tr_{\graph \Rightarrow \graphresult}: \graph \to \graphresult$ is the partial \LTS morphism defined by
\[
\tr_{\graph \Rightarrow \graphresult} (s) = \left\{
\begin{array}{ll}
m'_{\states{}}(g_{\states{}}(f^{-1}_{\states{}}(m^{-1}_{\states{}}(s))))             & {\rm if\ } f^{-1}_{\states{}}(m^{-1}_{\states{}}(s)) {\rm\ is\ defined} \\
\perp & {\rm otherwise}
\end{array}
\right.
\]
\end{definition}

The morphisms $m$, $m'$, $f$ and $g$ are as shown in Figure~\ref{fig:DPO}.
Track morphisms can be defined for sequences of direct transformations in a similar way, where for two direct transformations $\graph \Rightarrow \graphresult$ and $\graphresult \Rightarrow \graphresult'$, $\tr_{\graph \Rightarrow \graphresult \Rightarrow \graphresult'}(s)$ is defined as $\tr_{\graphresult \Rightarrow \graphresult'} \circ \tr_{\graph \Rightarrow \graphresult}(s)$.

\begin{definition}[Strong Joinability]
\label{def:strongjoin}
Given a transformation system $\Sigma$, a critical pair $\graphresult_0 \Leftarrow_{\rz, m_0} \graph \Rightarrow_{\ro, m_1} \graphresult_1$ is \emph{strongly joinable} if there are derivations $\graphresult_i \Rightarrow^*_\Sigma \Xgraph_i$, for $i = 0, 1$, an isomorphism $f: \Xgraph_0 \to \Xgraph_1$, and for each state $s \in \states{\graph}$, if both $\tr_{\graph \Rightarrow \graphresult_0}(s)$ and $\tr_{\graph \Rightarrow \graphresult_1}(s)$ are defined (that is, $s$ is persisting), then
\begin{enumerate}
\item $\tr_{\graph \Rightarrow \graphresult_0 \Rightarrow^*_\Sigma \Xgraph_0}(s)$ and $\tr_{\graph \Rightarrow \graphresult_1 \Rightarrow^*_\Sigma \Xgraph_1}(s)$ are defined;
\item $f_{\states{}}(\tr_{\graph \Rightarrow \graphresult_0 \Rightarrow^*_\Sigma \Xgraph_0}(s)) = \tr_{\graph \Rightarrow \graphresult_1 \Rightarrow^*_\Sigma \Xgraph_1}(s)$.
\end{enumerate}
\end{definition}

\begin{figure}[t]
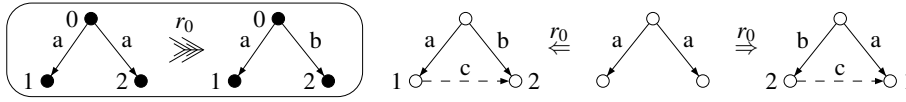

\centering
\scalebox{0.83}{
\begin{gpicture}(100,100)(0,0)

	\node[Nw=57,Nh=15](box)(40,15){}

	\node[Nw=2,Nh=2,linecolor=white](t1)(22,20){0}
	\node[Nw=2,Nh=2,fillcolor=black](p00)(25,20){}
	\node[Nw=2,Nh=2,fillcolor=black](p01)(33,10){}
	\node[Nw=2,Nh=2,fillcolor=black](p02)(18,10){}
	\node[Nw=2,Nh=2,linecolor=white](t1)(15,10){1}
	\node[Nw=2,Nh=2,linecolor=white](t1)(30,10){2}
	\drawedge(p00,p01){a}
	\drawedge[ELside=r](p00,p02){a}
	
	\drawline[AHnb=2,AHangle=30,AHLength=4,AHlength=0,AHdist=2,ATnb=2,ATangle=150,
					ATLength=2,ATlength=0,ATdist=2](40,15)(43,15)
	\node[Nw=2,Nh=2,linecolor=white](t1)(40,19){$\rz$}

	\node[Nw=2,Nh=2,linecolor=white](t1)(52,20){0}
	\node[Nw=2,Nh=2,fillcolor=black](p10)(55,20){}
	\node[Nw=2,Nh=2,fillcolor=black](p11)(48,10){}
	\node[Nw=2,Nh=2,fillcolor=black](p12)(63,10){}
	\node[Nw=2,Nh=2,linecolor=white](t1)(45,10){1}
	\node[Nw=2,Nh=2,linecolor=white](t1)(60,10){2}
	\drawedge(p10,p12){b}
	\drawedge[ELside=r](p10,p11){a}

	\node[Nw=2,Nh=2,linecolor=white](t1)(74,10){1}
	\node[Nw=2,Nh=2,linecolor=white](t1)(96,10){2}
	\node[Nw=2,Nh=2](s00)(85,20){}
	\node[Nw=2,Nh=2](s01)(77,10){}
	\node[Nw=2,Nh=2](s02)(93,10){}
	\drawedge[ELside=r](s00,s01){a}
	\drawedge(s00,s02){b}
	\drawedge[dash={1.5}0](s01,s02){c}

	\node[Nw=2,Nh=2,linecolor=white](t1)(100,15){$\Leftarrow$}
	\node[Nw=2,Nh=2,linecolor=white](t1)(100,17.5){$\rz$}

	\node[Nw=2,Nh=2](s00)(115,20){}
	\node[Nw=2,Nh=2](s01)(107,10){}
	\node[Nw=2,Nh=2](s02)(123,10){}
	\drawedge[ELside=r](s00,s01){a}
	\drawedge(s00,s02){a}

	\node[Nw=2,Nh=2,linecolor=white](t1)(130,15){$\Rightarrow$}
	\node[Nw=2,Nh=2,linecolor=white](t1)(130,17.5){$\rz$}

	\node[Nw=2,Nh=2,linecolor=white](t1)(134,10){2}
	\node[Nw=2,Nh=2,linecolor=white](t1)(156,10){1}
	\node[Nw=2,Nh=2](s00)(145,20){}
	\node[Nw=2,Nh=2](s01)(137,10){}
	\node[Nw=2,Nh=2](s02)(153,10){}
	\drawedge[ELside=r](s00,s01){b}
	\drawedge(s00,s02){a}
	\drawedge[dash={1.5}0](s01,s02){c}
	

\end{gpicture}
}
\caption{The need for strong joinability}
\label{fig:strongjoin}
\end{figure}

Def.~\ref{def:strongjoin} not only expresses that the \LTSs $\Xgraph_0$ and $\Xgraph_1$ need to be isomorphic, but besides that, that the states that \emph{persist} along direct transformations $\graph \Rightarrow_{\rz, m_0} \graphresult_0$, $\graph \Rightarrow_{\ro, m_1} \graphresult_1$, i.e.\ that are matched by glue-states of both $\rz$ and $\ro$, are in the end still present in both $\Xgraph_0$ and $\Xgraph_1$, and relatable to themselves. Consider the example in Figure~\ref{fig:strongjoin}. Rule $\rz$ can be applied in two ways on the given input. The results are isomorphic, but not in a bigger context (the dashed $c$-transition). To detect this, one should compare on which states the glue-states are matched.

Plump~\cite{plump.confluence2010} gives some suggestions how a transformation system can be equipped with a cover to determine whether a critical pair is strongly joinable. Based on that, we use the following approach: for a given critical pair, we define copies of $\rz$ and $\ro$ which we call $\rz^\kappa$ and $\ro^\kappa$, respectively, and we extend both $\Rlts^{\rz^\kappa}$ and $\Rlts^{\ro^\kappa}$ such that for each state $s$ in $\states{\Klts^{\rz^\kappa}}$ (and likewise in $\states{\Klts^{\ro^\kappa}}$), we add a self-loop transition with the fresh, unique label $\kappa$ to $g_{\states{}}(s)$. These selfloops, when the left pattern of $\rz^\kappa$ or $\ro^\kappa$ is matched on part of the conflict situation, are therefore introduced when applying a direct transformation, and then serve the purpose of marking the states that have been matched on glue-states. In this way, we can obtain \LTSs $\graphresult_0^\kappa$ and $\graphresult_1^\kappa$, i.e.\ the \LTSs $\graphresult_0$ and $\graphresult_1$ extended with the $\kappa$-selfloops. Once we have these, we relabel the $\kappa$-selfloops in $\graphresult_0^\kappa$ and $\graphresult_1^\kappa$, such that each state $s$ has a selfloop labelled $\kappa_s$, and after that, remove those $\kappa_s$-selfloops that do not appear in both \LTSs, i.e.\ that are not associated with $s$ in both $\graphresult_0^\kappa$ and $\graphresult_0^\kappa$. We call the resulting \LTSs $\graphresult_0^{\kappa_s}$ and $\graphresult_1^{\kappa_s}$.

Subsequent matches for rules $r \in \Sigma$ can only be established if they do not match a non-glue state on a persisting state, since trying to do so would violate the gluing conditions w.r.t.\ the related $\kappa_s$-selfloop. If it is detected that such a violating `match' can be made, then the critical pair is not strongly joinable, and the transformation fails. Besides that, each time a match has been established, we remove all $\kappa_s$-selfloops of states $s$ that have not been matched by any state.


\begin{algorithm}[t]
\caption{Conflict resolution algorithm}
\label{alg:confluence}
{\scriptsize
\begin{algorithmic}[2]
\REQUIRE Conflicts in set $C$
\ENSURE Returns $\FALSE$ iff there exists a conflict in $C$ that cannot be resolved, $\TRUE$ otherwise
\FORALL {$(\graphresult_0 \Leftarrow_{\rz, m_0} \conflict_f \Rightarrow_{\ro, m_1} \graphresult_1) \in C$}
	\STATE apply $\conflict_f \Rightarrow_{\rz^\kappa, m_0} \graphresult_0^\kappa$ and $\conflict_f \Rightarrow_{\ro^\kappa, m_1} \graphresult_1^\kappa$
	\STATE compute $\graphresult_0^{\kappa_s}$ and $\graphresult_1^{\kappa_s}$
	\STATE apply $\graphresult_0^{\kappa_s} \Rrightarrow^*_\Sigma \Xgraph_0$ and $\graphresult_1^{\kappa_s} \Rrightarrow^*_\Sigma \Xgraph_1$
	\IF {$\Xgraph_0 \not\simeq \Xgraph_1$ {\bf or transformation failed}}
		\RETURN $\FALSE$
	\ENDIF
\ENDFOR
\RETURN $\TRUE$
\end{algorithmic}
}
\end{algorithm}

In this way, states that persist along a sequence of direct transformations will still have their $\kappa_s$-selfloop in both $\Xgraph_0$ and $\Xgraph_1$. Then, it suffices to check that $\Xgraph_0$ and $\Xgraph_1$ are isomorphic.

In Alg.~\ref{alg:confluence}, our conflict resolution algorithm is presented, which can be used to determine, based on the constructed critical pairs, whether a transformation system is locally confluent, by establishing that all critical pairs are strongly joinable. With $\Rrightarrow$, we refer to applying a direct transformation after the removal of $\kappa_s$-selfloops of all the states that are not matched by states of the related transformation rule. The complexity of Alg.~\ref{alg:confluence} depends on the complexity of graph transformation, which is performed in lines 2 and 4, which in turn is dominated by the complexity of finding matches at line 4. In general, the graph matching problem~\cite{dodds.plump.graphtrans} is NP-complete. However, it has been shown in~\cite{dodds.plump.graphtrans} that if the graphs have a root, all states are reachable from that root, and each state has a bounded number $b$ of outgoing transitions, then the complexity is independent of the size of the input graph, instead only depending on $b$ and the number of transitions $n$ in the left pattern of the transformation rule. The complexity is then $\bigo(\Sigma_{i=0}^{n}b^i)$. Since our \LTSs are weakly connected, they meet these requirements. The other operations at lines 3 and 5 in Alg.~\ref{alg:confluence} can be performed in $\bigo(|\states{}|+|\transitions{}|)$, since they require scanning all states and transitions in the \LTSs once.


It has to be noted that if conflict detection is performed before resolution, all possible critical pairs need to be constructed. If instead, detection and resolution are mixed, i.e.\ each time a new critical pair is detected, it is immediately tested for resolvability, then non-confluent transformation systems can be identified as such as soon as a pair has been found that cannot be resolved. In practice, this means that the construction of all possible critical pairs can often be avoided.


Following, a proof sketch is given to show correctness of the technique. Consider a transformation system $\Sigma$ that is not confluent. Therefore, there exists an \LTS $\graph$ such that there are two direct transformations $\graphresult_0 \Leftarrow_{\rz, m_0} \graph \Rightarrow_{\ro, m_1} \graphresult_1$ that are not parallel independent. By Lemma~\ref{lem:parind}, this means that there must be at least one transition $x$ in $\graph$, say with label $a$, that is matched on by both $m_0$ and $m_1$, and at least one of the two rules removes $x$. Let $s$ and $t$ be the source states of the transitions $x_s$ and $x_t$ in $\Llts^\rz$ and $\Llts^\ro$ that match on $x$, respectively, and let $\rz$ define that $x$ must be removed. In Alg.~\ref{alg:conflict}, since $a \in \actions{\Llts^\rz \setminus \Klts^\rz} \cap \actions{\Llts^\ro}$, line 3 is skipped, and since $\Llts^\rz \not\simeq \Klts^\rz$, line 5 is skipped. We have $a \in \vectornot{\actions{\it out}}(s) \cap \actions{\it out}(t)$, hence at line 8, conflict compatibility morphisms will be computed with $f_{\states{}}(s) = t$. Since we only consider the largest possible conflict compatibility morphisms with a weakly connected domain of definition, we must have that $f_{\transitions{}}(x_s) = x_t$. If not, then either the source or target states in $\Llts^\rz$ and $\Llts^\ro$ are not relatable via $f$, which would mean that there is a gluing condition violation (Defs.~\ref{def:lts-transformation:transmatch} and \ref{def:conflict}), but that would mean that there can be no overlap of matches of $\rz$ and $\ro$ that involves $x_s$ and $x_t$. This would be in contradiction with the fact that there is a conflict between $\rz$ and $\ro$ involving $x$. By Def.~\ref{def:conflictsituation}, a conflict situation $\conflict_f = m_0(\Llts^\rz) \cup m_1(\Llts^\ro)$ is constructed at line 10 in Alg.~\ref{alg:conflict} with $m_{0,\transitions{}}(x_s)$ representing the overlap between $x_s$ and $x_t$. The subsequent inability to resolve the conflict using Alg.~\ref{alg:confluence} can be proven along the lines of the proof in~\cite{plump.confluence2010}.

The case that a given system is confluent can be proven as follows: in general, Alg.~\ref{alg:conflict} will produce some (possibly zero) conflicts. These conflicts, though, will be resolvable. This can be proven along the lines of the proof in~\cite{plump.confluence2010}.

\section{Conclusions}
\label{sec:conclusions}

In this paper, we discussed how conflicts in \LTS transformation systems can be efficiently detected and resolved. For the detection, we proposed a novel approach that tries to construct partial morphisms between the involved rule patterns. In particular cases, the absence of conflicts can be determined in linear time, for instance when one rule only removes transitions that another rule will never match on, because it does not refer to the particular transition label(s). This is a big improvement over previous approaches, like e.g.\ in~\cite{lambers.ehrig.orejas.2005}, since it is also applicable for two deleting rules, i.e.\ rules that remove transitions. For the resolution of conflicts, we have proposed an algorithm inspired by~\cite{plump.confluence2010}, but taylored to our particular setting using \LTSs. 
For future work, we will consider extensions, e.g.~\cite{lambers.ehrig.orejas.2006,heckel.confluence}, to extend our framework in comparable ways. Finally, for formal verification purposes, a hierarchy of different forms of confluence (ranging from strong to weak) has been identified concerning the behaviour described by LTSs~\cite{mateescu.wijs.confluence}. It would be interesting to see how these relate to confluence variants in the setting of graph and model transformation.









  \bibliographystyle{eptcs}
  \bibliography{bibliography}

\end{document}